\documentclass[runningheads]{llncs}
\usepackage[T1]{fontenc}
\usepackage{amsmath}
\usepackage{multirow}

\usepackage{graphicx}

\usepackage{subcaption}
\usepackage{amssymb}
\usepackage[inline]{enumitem}
\usepackage{hyperref}
\usepackage{cleveref}
\usepackage{xcolor}

\usepackage{tikz}
\usetikzlibrary{arrows,shapes,automata,petri,positioning}

\tikzset{
	place/.style={
		circle,
		thick,
		draw=blue!75,
		fill=blue!20,
		minimum size=6mm,
	},
	transitionH/.style={
		rectangle,
		thick,
		fill=black,
		minimum width=8mm,
		inner ysep=2pt
	},
	transitionV/.style={
		rectangle,
		thick,
		fill=black,
		minimum height=8mm,
		inner xsep=2pt
	}
}

\usepackage{forest}
\usepackage{tikz-qtree}

\newcommand{\twod}[1]{\mathbb{B}}
\newcommand{\threed}[1]{\mathbb{B}_{\star}}

\newcommand{\var}[1]{\text{var}_{{#1}}}
\newcommand{\IN}[1]{\text{IN}({#1})}
\newcommand{\ig}[1]{\text{G}({#1})}
\newcommand{\igf}[2]{\text{G}#1({#2})}
\newcommand{\astg}[1]{\text{astg}({#1})}

\begin{document}
\title{On the number of asynchronous attractors in AND-NOT Boolean networks}
\titlerunning{On the number of asynchronous attractors in AND-NOT Boolean networks}
\author{Van-Giang Trinh\inst{1} \and
Samuel Pastva\inst{2} \and
Jordan Rozum\inst{3} \and Kyu Hyong Park\inst{4} and R\'eka Albert\inst{4}}
\authorrunning{Trinh et al.}
\institute{Inria Saclay, EP Lifeware, Palaiseau, France\\
\email{van-giang.trinh@inria.fr}
\and
Faculty of Informatics, Masaryk University, Botanicka 68a, Brno, 60200, Czechia\\
\email{xpastva@fi.muni.cz}
\and
Department of Systems Science and Industrial Engineering, Binghamton University, Engineering Building, Vestal, 13850, New York, USA\\
\email{jrozum@binghamton.edu}
\and
Department of Physics, Pennsylvania State University, Davey Laboratory, University Park, 16802, Pennsylvania, USA\\
\email{\{kjp5774,rza1\}@psu.edu}}
\maketitle              % typeset the header of the contribution
\begin{abstract}
%The abstract should briefly summarize the contents of the paper in 150--250 words.

Boolean Networks (BNs) describe the time evolution of binary states using logic functions on the nodes of a network. They are fundamental models for complex discrete dynamical systems, with applications in various areas of science and engineering, and especially in systems biology.
A key aspect of the dynamical behavior of BNs is the number of attractors, which determines the diversity of long-term system trajectories. 
Due to the noisy nature and incomplete characterization of biological systems, a stochastic asynchronous update scheme is often more appropriate than the deterministic synchronous one.
AND-NOT BNs, whose logic functions are the conjunction of literals, are an important subclass of BNs because of their structural simplicity and their usefulness in analyzing biological systems for which the only information available is a collection of interactions among components.

In this paper, we establish new theoretical results regarding asynchronous attractors in AND-NOT BNs.
We derive two new upper bounds for the number of asynchronous attractors in an AND-NOT BN based on structural properties (strong even cycles and dominating sets, respectively) of the AND-NOT BN.
These findings contribute to a more comprehensive understanding of asynchronous dynamics in AND-NOT BNs, with implications for attractor enumeration and counting, as well as for network design and control.

\end{abstract}

\section{Introduction}\label{sec:introduction}

Boolean Networks (BNs) are discrete dynamical systems widely used to model complex interactions in various domains, ranging from sciences to engineering, notably systems biology~\cite{schwab2020concepts}. 
In these models, each variable represents a binary-state component that evolves over time according to an update scheme.
The dynamics of a BN are well specified by the set of Boolean update functions and the chosen update scheme, making it a powerful tool for studying regulatory and decision-making processes in real-world systems~\cite{Gates2021,rozum2021parity}.

The study of asynchronous attractors—the stable domains that arise under the asynchronous update scheme—is crucial for analyzing the stability and robustness of BNs~\cite{Saadatpour2010}.
Unlike the synchronous update scheme, wherein all variables are updated simultaneously, asynchronous updating introduces stochasticity by updating non-deterministically only one variable at each time step.
This update scheme is particularly relevant in biological systems and engineered systems, where components operate at different time-scales or respond to external signals asynchronously~\cite{Saadatpour2010,schwab2020concepts}.

A particularly significant subclass of BNs is the AND-NOT Boolean network class, in which each variable’s Boolean update function employs only AND (conjunction) and NOT (negation) connectives. 
Since in many biological systems the existing interaction information is given in a static way, as a signed directed graph (with activating or inhibitory arcs) on a set of components, this class of models is a way to capture the essential logical dependencies~\cite{richard2019positive,schwab2020concepts}.
Note that the influence graph of an AND-NOT BN uniquely specifies the set of Boolean update functions of this BN~\cite{https://doi.org/10.48550/arxiv.1211.5633,DBLP:journals/dam/RichardR13}.
Despite their structural simplicity, AND-NOT BNs exhibit rich dynamical behavior, including the emergence of attractors, which are fundamental in understanding long-term system evolution~\cite{https://doi.org/10.48550/arxiv.1211.5633,DBLP:journals/entcs/Veliz-CubaAL15}.
Several studies~\cite{DBLP:journals/bioinformatics/ParkLL10,https://doi.org/10.48550/arxiv.1211.5633,DBLP:journals/entcs/Veliz-CubaAL15} have shown that the class of AND-NOT BNs is sufficiently general for modeling biological dependencies and simple enough for theoretical analysis.
This class and its subtypes have attracted much attention from various research communities~\cite{Jarrah2010,DBLP:journals/dam/RichardR13,DBLP:journals/tcns/GaoCB18}.

A key goal in the study of BNs is determining the number of asynchronous attractors that a given BN can exhibit based on the structural properties of its influence graph.
Such results on the number of asynchronous attractors can have useful implications in understanding the diversity of possible long-term behaviors, as well as the stability, controllability, and computational complexity of the model considered~\cite{glass1973logical,thomas1973boolean}.
Specifically, knowing the upper and lower bounds on attractor numbers can aid the design of control-robust regulatory systems~\cite{Thomas1981,DBLP:journals/automatica/GaoCB18,Gates2021,rozum2021parity}.
Whereas previous research has explored theoretical results in general BNs, there remains a lack of rigorous theoretical results specifically addressing the quantitative limits on the number of asynchronous attractors in AND-NOT BNs.
See~\Cref{sec:related-work} for a more detailed summary of key contributions in this area.

In this work, we address the problem of relating the number of asynchronous attractors in an AND-NOT BN to the structural properties of the influence graph of this BN.
More specifically, we derive an upper bound for the number of asynchronous attractors based on the concept of strong even cycles, and a finer upper bound based on the concept of dominating sets. Both of these concepts address the capacity of a cycle to have multiple steady states.
Our results represent generalizations of the results on the number of fixed points presented in~\cite{https://doi.org/10.48550/arxiv.1211.5633}.

BNs have recently been connected to two formalisms in the field of symbolic AI, namely finite ground normal logic programs (see~\cite{trinh2024graphical}) and abstract argumentation frameworks (see~\cite{DBLP:conf/comma/DimopoulosD024,trinh2025graphical}).
The existing results on the number of fixed points of a BN have been applied to explore upper bounds for the number of stable models of a finite ground normal logic program or the number of stable extensions of an abstract argumentation framework. There are similar connections between asynchronous attractors and the number of regular models of a finite ground normal logic program as well as the number of preferred extensions of an abstract argumentation framework.
The encoded BN of an abstract argumentation framework is an AND-NOT BN~\cite{DBLP:conf/comma/DimopoulosD024,trinh2025graphical} and the encoded BN of a uni-rule finite ground normal logic program (an important subclass of normal logic programs~\cite{DBLP:conf/lpnmr/SeitzerS97}) is an AND-NOT BN~\cite{trinh2024graphical}.
This demonstrates the broader implications beyond systems biology of the results we obtain in this work.

The remainder of this paper is organized as follows. 
\Cref{sec:preliminaries} reviews relevant background on Boolean networks. 
\Cref{sec:related-work} briefly discusses key contributions in the area of relating network structure and dynamics.
\Cref{sec:strong-even-cycles-async-atts} and~\Cref{sec:domi-sets-async-atts} present the main results on attractor bounds.
Finally,~\Cref{sec:conclusion} concludes the paper and outlines future research directions.

\section{Preliminaries}\label{sec:preliminaries}

We use \(\twod{}\) to denote the Boolean domain \(\{0, 1\}\).
The Boolean connectives used in this paper include \(\land\) (AND) and \(\neg\) (NOT).

\subsection{Boolean networks}\label{subsec:pre-BN}

\begin{definition}\label{def:BN}
A Boolean Network (BN) \(f\) is a finite set of Boolean functions over a finite set of Boolean variables denoted by \(\var{f}\) with \(|f| = |\var{f}|\).
A variable \(v \in \var{f}\) is associated with a Boolean function \(f_v\), called the update function of \(v\).
Given a Boolean function \(f_v \in f\), \(\IN{f_v}\) denotes the set of input variables in \(f_v\).
We assume that \(\IN{f_v}\) only contains the variables that appear in \(f_v\) and essentially affect the evaluation of \(f_v\).
Then the signature of \(f_v\) is \(f_v \colon \twod{}^{\IN{f_v}} \to \twod{}\).
Function \(f_v\) is called \emph{constant} if either \(f_v = 0\) or \(f_v = 1\).
Variable \(v \in \var{f}\) is called a \emph{source variable} if \(f_v = v\).
\end{definition}

\begin{definition}\label{def:AND-NOT-BN}
A BN \(f\) is called an \emph{AND-NOT} BN if for every \(v \in \var{f}\), \(f_v\) is either a constant or a conjunction of literals (i.e., variables or their negations connected by logical-and rules).
We assume that \(f_v\) does not contain two literals of the same variable.
\end{definition}

\subsection{Dynamics and attractors}

\begin{definition}\label{def:BN-state}
Given a BN \(f\),
a state \(s\) of \(f\) is a Boolean vector \(s \in \mathbb{B}^{|\var{f}|}\).
State \(s\) can be seen as a mapping, \(s \colon \var{f} \to \mathbb{B}\).
We write \(s_v\) or \(s(v)\) to denote the value of variable \(v\) in \(s\).
For convenience, we write a state simply as a string of values of variables in this state (e.g., we write 0110 instead of (0, 1, 1, 0)).
\end{definition}

Given a BN \(f\),
an update scheme specifies the way the variables in \(f\) update their states at each time step.
There are two major types of update schemes~\cite{schwab2020concepts}: synchronous (in which all variables update simultaneously) and asynchronous (in which a single variable is non-deterministically chosen for updating).
In the present work, we focus on the asynchronous update scheme.
Under this update scheme, the dynamics of the BN are represented by a directed graph, called the \emph{asynchronous state transition graph}, and defined formally in~\Cref{def:asynchronous-stg}.

\begin{definition}
Given a directed graph \(G\), \(V(G)\) (resp.\ \(E(G)\)) denotes the set of vertices (resp.\ arcs) of \(G\).
\(G[B]\) denotes the induced graph w.r.t. \(B \subseteq V(G)\) of \(G\).
\end{definition}

\begin{definition}\label{def:asynchronous-stg}
Given a BN \(f\),
the \emph{asynchronous state transition graph} of \(f\) (denoted by \(\astg{f}\)) is given as follows.
\(V(\astg{f}) = \twod{}^{|\var{f}|}\) and \((x, y) \in E(\astg{f})\) iff there exists a variable \(v_i \in \var{f}\) such that \(y(v_i) = f_i(x) \neq x(v_i)\) and \(y(v_j) = x(v_j)\) for all \(v_j \not= v_i\).
\end{definition}

An \emph{attractor} of a BN is defined as a subset-minimal \emph{trap set}, which depends on the employed update scheme.
Equivalently, an attractor is a terminal strongly connected component of the state transition graph corresponding to the employed update scheme~\cite{DBLP:journals/dam/Richard09}. 
To ease the statement of our results, we define the concepts of trap set and attractor for directed graphs in general.

\begin{definition}\label{def:digraph-attractor}
Given a directed graph \(G\),
a trap set \(A\) of \(G\) is a non-empty subset of \(V(G)\) such that there is no arc in \(G\) going out of \(A\) (formally, there do not exist two vertices \(x \in A\) and \(y \in V(G) \setminus A\) such that \((x, y) \in E(G)\)).
An \emph{attractor} of \(G\) is defined as a \(\subseteq\)-minimal trap set of \(G\).
An attractor is called a \emph{fixed point} if it consists of only one vertex, and a \emph{cyclic attractor} otherwise.
Equivalently, \(A\) is an attractor of \(G\) iff \(A\) is a terminal strongly connected component of \(G\).
\end{definition}

\begin{definition}\label{def:BN-asynchronous-attractor}
Given a BN \(f\),
an asynchronous attractor of \(f\) is defined as an attractor of \(\astg{f}\).
\end{definition}

\begin{example}\label{exam:BN-dynamics-attractor}
Consider the AND-NOT BN \(f\) given as \(\var{f} = \{a, b, c\}\) and \(f_a = \neg b \land c\), \(f_b = \neg a \land \neg c\), and \(f_c = \neg a\).
The asynchronous state transition graph
\(\astg{f}\) is shown in~\Cref{fig:exam-BN-astg-gig}~(a).
The graph \(\astg{f}\) has a unique cyclic attractor: \{000, 010, 011, 001, 101, 100\}.
\end{example}

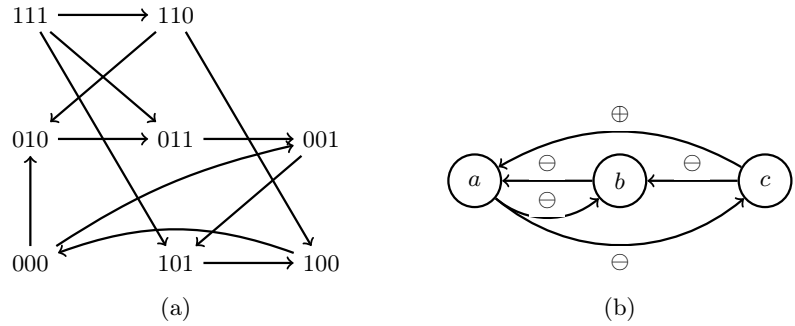
\begin{figure}
\centering
    \begin{subfigure}{.55\linewidth}
    \centering
        \begin{tikzpicture}[node distance=1.2cm and 1.2cm, every node/.style={scale=1.0}, line width = 0.3mm]
            \node[] (7) [] {111};
	    \node[] (6) [right=of 7] {110};
	    \node[] (2) [below=of 7] {010};
	    \node[] (3) [right=of 2] {011};
            \node[] (1) [right=of 3] {001};
            \node[] (0) [below=of 2] {000};
            \node[] (5) [right=of 0] {101};
            \node[] (4) [right=of 5] {100};

            \draw[->] (0) edge [] (2);
            \draw[->] (0) edge [bend left=10] (1);

            \draw[->] (1) edge [] (5);
            
            \draw[->] (2) edge [] (3);

            \draw[->] (3) edge [] (1);

            \draw[->] (4) edge [bend right=20] (0);

            \draw[->] (5) edge [] (4);
            
            \draw[->] (6) edge [] (2);
            \draw[->] (6) edge [] (4);

            \draw[->] (7) edge [] (6);
            \draw[->] (7) edge [] (5);
            \draw[->] (7) edge [] (3);
        \end{tikzpicture}
        \caption{}
    \end{subfigure}
    \begin{subfigure}{.4\linewidth}
    \centering
        \begin{tikzpicture}[node distance=1.2cm and 1.2cm, every node/.style={scale=1.0}, line width = 0.3mm]
            \node[circle, draw, minimum size=7mm] (a) [] {$a$};
            \node[circle, draw, minimum size=7mm] (b) [right=of a] {$b$};
            \node[circle, draw, minimum size=7mm] (c) [right=of b] {$c$};
		
            \draw[->] (a) edge [bend right=40] node [midway, below, fill=white] {$\ominus$} (c);
            \draw[->] (a) edge [bend right=40] node [midway, above, fill=white] {$\ominus$} (b);
            
            \draw[->] (b) edge [] node [midway, above, fill=white] {$\ominus$} (a);
            
            \draw[->] (c) edge [] node [midway, above, fill=white] {$\ominus$} (b);
            \draw[->] (c) edge [bend right=30] node [midway, above, fill=white] {$\oplus$} (a);
        \end{tikzpicture}
        \caption{}
    \end{subfigure}
    \caption{(a) asynchronous state transition graph and (b) global influence graph of the BN given in~\Cref{exam:BN-dynamics-attractor}}
\label{fig:exam-BN-astg-gig}
\end{figure}

\subsection{Influence graphs}

In this subsection we present the process of obtaining the influence graph of a BN, define different types of cycles in the influence graph, and introduce a network pattern via which a vertex connects to a cycle in an incoherent way.

Let \(x\) be a state of \(f\) and consider the binary value \(a \in \mathbb{B}\).
We use \(x[v \leftarrow a]\) to denote the state \(y\) such that \(y_v = a\) and \(y_u = x_u, \forall u \in \var{f}, u \neq v\). In other words, \(x[v \leftarrow a]\) has the value \(a\) substituted for the variable \(v\).

\begin{definition}\label{def:local-ig}
Given a BN \(f\) and a state \(x\) of \(f\),
the \emph{local influence graph} of \(f\) w.r.t. \(x\) (denoted by \(\igf{f}{x}\) and consisting of signed arcs) is defined as follows:
\begin{itemize}
    \item \((v_jv_i, \oplus)\) is an arc of \(\igf{f}{x}\) iff \(f_{v_i}(x[v_j \leftarrow 0]) < f_{v_i}(x[v_j \leftarrow 1])\)
    \item \((v_jv_i, \ominus)\) is an arc of \(\igf{f}{x}\) iff \(f_{v_i}(x[v_j \leftarrow 0]) > f_{v_i}(x[v_j \leftarrow 1])\)
\end{itemize}
\end{definition}

\begin{definition}\label{def:global-ig}
Given a BN \(f\),
the \emph{global influence graph} of \(f\) (denoted by \(\ig{f}\)) is defined as \(\ig{f} = \cup_{x \in \mathbb{B}^{|\var{f}|}}\igf{f}{x}\).
\end{definition}

\begin{definition}\label{def:digraph-even-odd-cycle}
Given a signed directed graph \(G\),
a cycle of \(G\) is called \emph{even} (resp.\ \emph{odd}) iff it contains an even (resp.\ odd) number of negative arcs.
An even (resp.\ odd) \emph{feedback vertex set} of \(G\) is a (possibly empty) subset of \(V(G)\) that intersects every even (resp.\ odd) cycle of \(G\).
\end{definition}

\begin{definition}[\cite{DBLP:journals/dam/RichardR13}]\label{def:delocalizing-triple}
    Consider a signed directed graph \(G\), a cycle \(C\) of \(G\), and vertices \(u\), \(v_1\), \(v_2\) of \(G\).
    \((u, v_1, v_2)\) is said to be a \emph{delocalizing triple} of \(C\) when 1) \(v_1\), \(v_2\) are distinct vertices of \(C\); 2) \((uv_1, \oplus)\) and \((uv_2, \ominus)\) are arcs of \(G\) but not arcs of \(C\).
    This delocalizing triple is \emph{internal} if \(u\) is a vertex of \(C\) and \emph{external} otherwise.
\end{definition}

\begin{example}\label{exam:BN-influence-graph-cycle}
Consider again the BN \(f\) given in~\Cref{exam:BN-dynamics-attractor}.
The global influence graph \(\ig{f}\) is shown in~\Cref{fig:exam-BN-astg-gig}~(b).
This graph has an even cycle \(C_1 = a \xrightarrow{\ominus} b \xrightarrow{\ominus} a\), an odd cycle \(C_2 = a \xrightarrow{\ominus} c \xrightarrow{\oplus} a\), and an odd cycle \(C_3 = a \xrightarrow{\ominus} c \xrightarrow{\ominus} b \xrightarrow{\ominus} a\).
Following~\Cref{def:delocalizing-triple}, \(C_1\) has an external delocalizing triple (namely \((c, a, b)\)), whereas \(C_2\) and \(C_3\) have no delocalizing triples.
\end{example}

\section{Related work}\label{sec:related-work}

The relationship between the structure of a BN (such as its global/local influence graph) and its dynamical behavior (such as its attractors) has been a central focus of research over the past decades. 
Below, we discuss key contributions in this area for general BNs and AND-NOT BNs in particular.

\subsection{General Boolean networks}

It was shown early on that a BN without cycles has a unique asynchronous attractor that coincides with a unique fixed point~\cite{Robert1980}.
Even and odd cycles in the global influence graph have been considered starting with two conjectures by Ren\'e Thomas~\cite{Thomas1981}.
Several results have been obtained regarding this direction: the fixed points coincide with the asynchronous attractors under the absence of odd cycles~\cite{DBLP:journals/aam/Richard10}, the number of fixed points is at most one under the absence of even cycles~\cite{Aracena2008}, the number of fixed points or asynchronous attractors is at most \(2^{|U|}\) where \(U\) is an even feedback vertex set of the global influence graph~\cite{Aracena2008,DBLP:journals/dam/Richard09}.
Other upper bounds for the number of fixed points have been explored based on the connections of BNs with Coding Theory~\cite{Gadouleau2015}.

By considering local influence graphs, several stronger results have been obtained.
The work by~\cite{DBLP:journals/aam/ShihD05} showed that if the local influence graph w.r.t. any state has no cycles, then the BN has a unique fixed point.
This result was generalized later in~\cite[Corollary 3]{DBLP:journals/tcs/Richard15}, showing that the BN still has a unique fixed point if, for any \(l\), there are fewer than \(2^l\) states \(x\) such that the local influence graph w.r.t. \(x\) has a cycle of length at most \(l\).
Inspired by this result, the work by~\cite{DBLP:journals/aam/RemyRT08} showed that if the local influence graph w.r.t. any state lacks even cycles, then the BN has a unique fixed point.
The absence of even cycles in every local influence graph also implies the uniqueness of asynchronous attractors~\cite{DBLP:journals/dam/RichardC07}.
It is then natural to ask the question for odd cycles: Is it true that if the local influence graph w.r.t. any state has no odd cycles, then the asynchronous attractors coincide with the fixed points?
However, several counterexamples using different techniques have shown that the answer to this question is `No'~\cite{DBLP:journals/dam/Ruet17,DBLP:journals/siamads/TonelloFC19}. 

\subsection{AND-NOT Boolean networks}

One interesting and useful direction is to seek similar or stronger results for subclasses of BNs~\cite{DBLP:journals/dam/RichardR13,https://doi.org/10.48550/arxiv.1211.5633,DBLP:journals/entcs/Veliz-CubaAL15,DBLP:journals/tcns/GaoCB18}.
Indeed, several results have been obtained for AND-NOT BNs. The fact that multiple inputs converging into a vertex are connected with AND in these networks allows the identification of influence patterns that take away a cycle's potential for multi-stability.
The work by~\cite{DBLP:journals/dam/RichardR13} showed that if the global influence graph has no even cycles without delocalizing triples, then the AND-NOT BN has at most one fixed point.
This article also showed that if every odd cycle of the global influence graph has an internal delocalizing triple, then the AND-NOT BN has at least one fixed point.
Although sharing the common result on even cycles with~\cite{DBLP:journals/dam/RichardR13}, the work by~\cite{https://doi.org/10.48550/arxiv.1211.5633} went further by relating the number of fixed points to delocalizing triples, then to dominating sets.
The results of~\cite{https://doi.org/10.48550/arxiv.1211.5633} are stronger than the respective result for fixed points based on even feedback vertex sets~\cite{Aracena2008}.

\section{Relating the number of asynchronous attractors to strong even cycles}\label{sec:strong-even-cycles-async-atts}

The number of fixed points in AND-NOT BNs has been studied in~\cite{https://doi.org/10.48550/arxiv.1211.5633}.

\begin{theorem}[Theorem 3.5 of~\cite{https://doi.org/10.48550/arxiv.1211.5633}]\label{theo:AND-NOT-BN-num-fix-cycle-delo-triple}
    Given an AND-NOT BN \(f\),
    let \(U\) be a subset of \(\var{f}\) that intersects every delocalizing-triple-free even cycle of \(\ig{f}\).
    Then the number of fixed points of \(f\) is at most \(2^{|U|}\).
\end{theorem}

Inspired by this result and the fact that a fixed point is a special asynchronous attractor, we here address the number of asynchronous attractors of an AND-NOT BN.
\Cref{theo:AND-NOT-BN-num-async-att-strong-even-cycle} bounds the number of asynchronous attractors using the subset of \(\var{f}\) that intersects every strong (i.e., delocalizing-triple-free) even cycle of \(\ig{f}\). 
We build up to~\Cref{theo:AND-NOT-BN-num-async-att-strong-even-cycle} by showing that an AND-NOT BN whose global influence graph has no strong cycles has a unique asynchronous attractor.

\begin{definition}\label{def:strong-cycle}
    Given a signed directed graph \(G\), a cycle \(C\) of \(G\).
    The cycle \(C\) is \emph{strong} if it has no delocalizing triples.
\end{definition}

\begin{definition}[\cite{DBLP:journals/dam/RichardR13}]\label{def:local-cycle-ig}
    Given a BN \(f\),
    a cycle of \(\ig{f}\) is called \emph{local} iff there is a state \(x\) such that the cycle also exists in the local influence graph of \(x\), \(\igf{f}{x}\).
\end{definition}

\begin{proposition}\label{prop:AND-NOT-BN-local-cycle-strong}
    Let \(f\) be an AND-NOT BN.
    If \(C\) is a local cycle of \(\ig{f}\), then it must be strong in \(\ig{f}\).
\end{proposition}
\begin{proof}
    Assume that \(C\) is not strong in \(\ig{f}\).
    Then there exist three vertices, \(v_i, v_j \in V(C), v_i \neq v_j\), as well as \(v_k \in V(\ig{f})\), such that \(v_k \xrightarrow{\oplus} v_i\) and \(v_k \xrightarrow{\ominus} v_j\) are arcs of \(\ig{f}\) but not arcs of \(C\).
    By the definition of AND-NOT BNs, this property implies that \(f_{v_i} = v_k \land \dots\) and \(f_{v_j} = \neg v_k \land \dots\).
    
    Since \(C\) is local, there is a state \(x\) such that \(C\) is a cycle of \(\igf{f}{x}\).
    We have two cases as follows.
    \textbf{Case 1}: \(x_{v_k} = 0\).
    Then \(f_{v_i}(x[u \leftarrow 0]) = f_{v_i}(x[u \leftarrow 1]) = 0\) for all \(u \in \IN{f_{v_i}} \setminus \{v_k\}\).
    Since \(f_{v_i}\) is a constant or a conjunction of literals, there is at most one arc from \(v_k\) to \(v_i\).
    Hence, if \(\igf{f}{x}\) has an arc ending at \(v_i\), this arc must start from \(v_k\), and must be \((v_kv_i, \oplus)\).
    Since \((v_kv_i, \oplus)\) is not an arc of \(C\) and \(\igf{f}{x}\) has no other arc ending at \(v_i\), no arc of \(C\) ends at \(v_i\), which is a contradiction because \(v_i\) is a vertex of \(C\).
    It is easy to see that the condition that \(v_k \xrightarrow{\oplus} v_i\) and \(v_k \xrightarrow{\ominus} v_j\) are not arcs of \(C\) is necessary for the proof.
    \textbf{Case 2}: \(x_{v_k} = 1\).
    Then \(f_{v_j}(x[u \leftarrow 0]) = f_{v_j}(x[u \leftarrow 1]) = 0\) for all \(u \in \IN{f_{v_j}} \setminus \{v_k\}\).
    Since \(f_{v_j}\) is a constant or a conjunction of literals, there is at most one arc from \(v_k\) to \(v_j\).
    Hence, if \(\igf{f}{x}\) has an arc ending at \(v_j\), this arc must start from \(v_k\), and must be \((v_kv_j, \ominus)\).
    Since \((v_kv_j, \ominus)\) is not an arc of \(C\) and \(\igf{f}{x}\) has no other arc ending at \(v_j\), no arc of \(C\) ends at \(v_j\), which is a contradiction because \(v_j\) is a vertex of \(C\).
    
    Both cases lead to a contradiction.
    Hence, \(C\) must be strong in \(\ig{f}\).\qed%
\end{proof}

\begin{theorem}[Corollary 1 of~\cite{DBLP:journals/dam/RichardC07}]\label{theo:BN-no-even-cycle-local-ig-unique-async-att}
    Consider a BN \(f\) and suppose that \(\igf{f}{x}\) has no even cycles for every \(x \in \twod{}^{|\var{f}|}\).
    Then \(\astg{f}\) has a unique attractor.
\end{theorem}

\begin{lemma}\label{lem:AND-NOT-BN-no-even-cycle-unique-async-att}
    Given an AND-NOT BN \(f\),
    if \(\ig{f}\) lacks strong even cycles, then \(\astg{f}\) has a unique attractor.
\end{lemma}
\begin{proof}
    Assume that \(\astg{f}\) has two attractors.
    By~\Cref{theo:BN-no-even-cycle-local-ig-unique-async-att}, there must be a state \(x\) such that \(\igf{f}{x}\) has an even cycle (say \(C\)).
    Then \(C\) is a local cycle of \(\ig{f}\) since \(\igf{f}{x}\) is a sub-graph of \(\ig{f}\).
    By~\Cref{prop:AND-NOT-BN-local-cycle-strong}, \(C\) must be strong, which is a contradiction.
    Hence, \(\astg{f}\) has at most one attractor.
    Since \(\astg{f}\) has at least one attractor, it has a unique attractor. \qed%
\end{proof}

We now introduce a mapping between the attractors of a BN and its variants in which certain variables are made into source variables. 
This mapping will be used in the proof of~\Cref{theo:AND-NOT-BN-num-async-att-strong-even-cycle}.

\begin{proposition}\label{prop:att-sub-graph}
    Given a directed graph \(G\).
    Let \(G'\) be a sub-graph of \(G\) such that \(G\) and \(G'\) have the same set of vertices.
    There is an injection (one-to-one mapping) from the set of attractors of \(G\) to that of \(G'\).
    Consequently, the number of attractors of \(G\) is less than or equal to that of \(G'\).
\end{proposition}
\begin{proof}
    Let \(\mathcal{A}(G)\) denote the set of attractors of a directed graph \(G\).
    We construct a mapping \(m \colon \mathcal{A}(G) \to \mathcal{A}(G')\) as follows.
    Consider attractor \(B \in \mathcal{A}(G)\).
    The corresponding induced sub-graphs in \(G\) and \(G'\) have the property that \(E(G'[B]) \subseteq E(G[B])\).
    If \(E(G'[B]) = E(G[B])\), then choose \(m(B) = B\).
    In the case that \(E(G'[B]) \subsetneq E(G[B])\), \(B\) may not be an attractor of \(G'\) but it is a trap set of \(G'\) because \(E(G') \subseteq E(G)\) ensures that the terminal nature of \(B\) is preserved.
    Then \(B\) contains at least one \(\subseteq\)-minimal trap set (equivalently an attractor) of \(G'\), say \(B'\).
    Choose \(m(B) = B'\).
    By construction, \(m(B) \subseteq B\) for any \(B \in \mathcal{A}(G)\).
    Consider \(B_1, B_2 \in \mathcal{A}(G)\) and \(B_1 \neq B_2\).
    We have \(m(B_1) \subseteq B_1\) and \(m(B_2) \subseteq B_2\).
    Since \(B_1 \cap B_2 = \emptyset\), \(m(B_1) \cap m(B_2) = \emptyset\).
    This implies that \(m\) is an injection.
    Consequently, \(|\mathcal{A}(G)| \leq |\mathcal{A}(G')|\). \qed%
\end{proof}

\begin{lemma}\label{lem:BN-one-source-node-sub-graph}
    Let \(f\) be a BN.
    Let \(g\) be the BN that differs from \(f\) in that it makes a single variable into a source variable, i.e., \(\var{f} = \var{g}\), \(g_{v_j} = f_{v_j}, \forall v_j \neq v_i\) and \(g_{v_i} = v_i\) where \(v_i\) is a given variable in \(\var{f}\).
    Then \(V(\astg{g}) = V(\astg{f})\) and \(E(\astg{g}) \subseteq E(\astg{f})\).
\end{lemma}
\begin{proof}
    Since \(\var{f} = \var{g}\), \(V(\astg{g}) = V(\astg{f})\).
    Consider an arc \((x, y) \in E(\astg{g})\).
    By definition, there is a variable \(v_j \in \var{g}\) such that \(y(v_j) = g_{v_j}(x) \neq x(v_j)\). In other words, the \((x, y)\) arc corresponds to the update of variable \(v_j\).
    We have two cases as follows.
    \textbf{Case 1}: \(v_j = v_i\). 
    This is impossible because in this case, \(g_{v_j}(x) = x(v_j)\) because \(v_j\) is a source variable of \(g\) by construction.
    \textbf{Case 2}: \(v_j \neq v_i\).
    Then \(y(v_j) = g_{v_j}(x) = f_{v_j}(x) \neq x(v_j)\).
    By definition, \((x, y) \in E(\astg{f})\). In other words, this update is shared by \(g\) and \(f\).
    Now we can conclude that \(E(\astg{g}) \subseteq E(\astg{f})\). \qed%
\end{proof}

We are now ready to prove the main result of this section.

\begin{theorem}\label{theo:AND-NOT-BN-num-async-att-strong-even-cycle}
    Consider an AND-NOT BN \(f\).
    Let \(U\) be a subset of \(\var{f}\) that intersects every strong even cycle of \(\ig{f}\).
    Then the number of attractors of \(\astg{f}\) is at most \(2^{|U|}\).
\end{theorem}
\begin{proof}
    Let \(g\) be the BN such that \(g_{v_i} = v_i\) for every \(v_i \in U\) and \(g_{v_i} = f_{v_i}\) for every \(v_i \in \var{f} \setminus U\).
    By applying~\Cref{lem:BN-one-source-node-sub-graph} sequentially over variables in \(U\), we obtain \(V(\astg{g}) = V(\astg{f})\) and \(E(\astg{g}) \subseteq E(\astg{f})\).
    
    By construction, \(\ig{g}\) breaks all strong even cycles in \(\ig{f}\) because \(U\) intersects every strong even cycle in \(\ig{f}\), but adds self-arcs (which are strong even cycles by definition) for variables in \(U\) (because \(g_{v_i} = v_i\) for every \(v_i \in U\)).
    Consider an even cycle \(C\) of \(\ig{f}\) that is not strong.
    Let \((t_C, v_1, v_2)\) be a delocalizing triple of \(C\).
    Assume that \(C\) becomes strong in \(\ig{g}\).
    This can happen if the arc from \(t_C\) to \(v_1\) or the arc from \(t_C\) to \(v_2\) is removed in the transformation from \(f\) to \(g\).
    However, since in this transformation all input arcs of \(v_i \in U\) are removed from \(\ig{f}\), the vertex affected by the transformation will lose all of its input arcs, including the arc inside the cycle, leading to \(C\) being broken in \(\ig{g}\), which is a contradiction.
    This implies that \(C\) cannot become strong in \(\ig{g}\).
    Hence, \(\ig{g}\) contains no strong even cycles apart from the self-even-cycles corresponding to the variables in \(U\).
    
    For each Boolean assignment \(x\) w.r.t. variables in \(U\), we build a new BN \(g^x\) such that \(\var{g^x} = \var{f}\), \(g^x_{v_i} = x(v_i)\) for every \(v_i \in U\) and \(g^x_{v_i} = g_{v_i}\) for every \(v_i \in \var{f} \setminus U\).
    It is easy to see that the attractors of \(\astg{g^x}\) are identical to the attractors of \(\astg{g}\) that agree with \(x\), since all the variables in \(U\) are source variables of \(g\) by construction.
    Since \(f\) is an AND-NOT BN, \(g\) and \(g^x\) are clearly AND-NOT BNs.
    By construction, \(\ig{g^x}\) breaks all self-cycles of \(\ig{g}\), thus \(\ig{g^x}\) has no strong even cycles.
    By~\Cref{lem:AND-NOT-BN-no-even-cycle-unique-async-att}, \(\astg{g^x}\) has a unique attractor.

    There are \(2^{|U|}\) possible Boolean assignments w.r.t. variables in \(U\).
    This implies that \(\astg{g}\) has at most \(2^{|U|}\) attractors.
    By~\Cref{prop:att-sub-graph}, \(\astg{f}\) has at most \(2^{|U|}\) attractors. \qed%
\end{proof}

\begin{example}
    Consider again the AND-NOT BN \(f\) given in~\Cref{exam:BN-dynamics-attractor}.
    The global influence graph \(\ig{f}\) has only one even cycle, namely (\(C_1 = a \xrightarrow{\ominus} b \xrightarrow{\ominus} a\)), so the size of the minimal even feedback vertex set is 1.
    By the upper bound proved in~\cite{DBLP:journals/dam/Richard09} (based on even feedback vertex sets), \(\astg{f}\) has at most 2 attractors.
    Even cycle \(C_1\) has a delocalizing triple (namely \((c, a, b)\)).
    By~\Cref{theo:AND-NOT-BN-num-fix-cycle-delo-triple}, \(f\) has at most 1 fixed point.
    Indeed, \(f\) has no fixed point.
    By applying~\Cref{theo:AND-NOT-BN-num-async-att-strong-even-cycle}, we obtain \(\astg{f}\) has at most 1 attractor.
    Indeed, \(\astg{f}\) has a unique cyclic attractor.
\end{example}
\section{Relating the number of asynchronous attractors to dominating sets}\label{sec:domi-sets-async-atts}

The number of fixed points in AND-NOT BNs has also been connected to the size of dominating sets in~\cite{https://doi.org/10.48550/arxiv.1211.5633}. The key intuition is that certain patterns of influence incident on an even cycle take away the cycle's potential for multi-stability. Such cycles, referred to as inconsistent, do not contribute to the diversity of fixed points. The concept of dominating set was introduced to refer to the determinant of fixed point diversity in AND-NOT BNs. 

\begin{definition}[\cite{https://doi.org/10.48550/arxiv.1211.5633}]\label{def:consistent-cycle}
    A cycle \(C\) of a signed directed graph \(W\) is called \emph{inconsistent} if there is a vertex \(k_C\) such that there is a positive path of the form \(k_C \xrightarrow{\oplus} i_1 \xrightarrow{\oplus} \dots \xrightarrow{\oplus} i_r \xrightarrow{\oplus} t_C\) and a negative path of the form \(k_C \xrightarrow{\oplus} j_1 \xrightarrow{\oplus} \dots \xrightarrow{\oplus} j_r \xrightarrow{\ominus} u_C\), such that: %\(t_C\) and \(u_C\) are distinct vertices in \(C\); \(k_C \xrightarrow{\oplus} t_C\) and \(k_C \xrightarrow{\ominus} u_C\) are not arcs in \(C\); \(|I_{j_1}| = \dots = |I_{j_r}| = 1\) where \(I_{j}\) is the set of input arcs of vertex \(j\) in \(W\).
    \begin{itemize}
	\item \(t_C\) and \(u_C\) are distinct vertices of \(C\).
	\item \(k_C \xrightarrow{\oplus} t_C\) and \(k_C \xrightarrow{\ominus} u_C\) are not arcs of \(C\).
	\item \(|I_{j_1}| = \dots = |I_{j_r}| = 1\) where \(I_{j}\) is the set of input arcs of vertex \(j\) in \(W\).
    \end{itemize}
    If such a vertex \(k_C\) does not exist, \(C\) is called \emph{consistent}.
\end{definition}

To see in what way an inconsistent cycle loses multi-stability, consider that \(k_C\) has a fixed value. The requirement that each mediator of the negative path between  \(k_C\)  and  \(u_C\) has a single input ensures that all the mediators adopt the same value as \(k_C\). Thus, one of the inputs to \(u_C\) is the value opposite of \(k_C\), which connects with AND to the value of a within-cycle input. As a consequence, if the fixed value of \(k_C\) is 1, \(u_C\) will adopt the value of 0, regardless of the within-cycle input, which does not allow multi-stability. If the fixed value of \(k_C\) is 0, all the members of the positive path from \(k_C\) to \(t_C\) will adopt the state 0. 

\begin{definition}[\cite{https://doi.org/10.48550/arxiv.1211.5633}]\label{def:dominating-set}
    A subset of vertices \(J\) \emph{dominates} a signed directed graph \(W\) if
    \begin{itemize}
	\item \(J\) intersects every consistent even cycle of \(W\);
	\item for each even cycle \(C\) that has no delocalizing triples and is inconsistent, \(J\) intersects \(C\) or \(J\) contains at least one \(k_C\) that makes \(C\) inconsistent.
    \end{itemize}
\end{definition}

The intuition behind the dominating set is that fixing the value of all vertices in a dominating set ensures the fixed value of a vertex for every even cycle.

\begin{theorem}[Theorem 3.7 of~\cite{https://doi.org/10.48550/arxiv.1211.5633}]\label{theo:AND-NOT-BN-num-fix-domi-set}
    Let \(f\) be an AND-NOT BN.
    Let \(U\) be a subset of \(\var{f}\) that dominates \(\ig{f}\).
    Then \(f\) has at most \(2^{|U|}\) fixed points.
\end{theorem}

We revealed that the positive path and the negative path in~\Cref{def:consistent-cycle} do not need to have the same length.
Moreover, the proof of~\Cref{theo:AND-NOT-BN-num-fix-domi-set} omitted the case that a Boolean function becomes constant when adding a new arc to the global influence graph (see more details in Appendix B of~\cite{https://doi.org/10.48550/arxiv.1211.5633}).
Hereafter, we revise the definition of consistent cycles (\Cref{def:consistent-cycle-new}), but instead of fixing the missing part of the proof, we prove a stronger result on the number of asynchronous attractors of an AND-NOT BN (\Cref{theo:AND-NOT-BN-num-async-att-domi-set}).

\begin{definition}\label{def:consistent-cycle-new}
    A cycle \(C\) of a signed directed graph \(W\) is called \emph{inconsistent} if there is a vertex \(k_C\) such that there is a positive path of the form \(k_C \xrightarrow{\oplus} i_1 \xrightarrow{\oplus} \dots \xrightarrow{\oplus} i_r \xrightarrow{\oplus} t_C\) (\(r \geq 0\)) and a negative path of the form \(k_C \xrightarrow{\oplus} j_1 \xrightarrow{\oplus} \dots \xrightarrow{\oplus} j_m \xrightarrow{\ominus} u_C\) (\(m \geq 0\)), such that:
    \begin{itemize}
	\item \(t_C\) and \(u_C\) are distinct vertices of \(C\).
        \item \(k_C \xrightarrow{\oplus} t_C\) and \(k_C \xrightarrow{\ominus} u_C\) are not arcs of \(C\).
	\item If \(m > 0\), then \(|I_{j_1}| = \dots = |I_{j_m}| = 1\) where \(I_{j}\) is the set of input arcs of vertex \(j\) in \(W\).
    \end{itemize}
    If such a vertex \(k_C\) does not exist, \(C\) is called \emph{consistent}.
\end{definition}

The proof of~\Cref{theo:AND-NOT-BN-num-async-att-domi-set} relies on the concept of \emph{percolation} in a BN.

\begin{definition}\label{def:BN-one-step-percolation}
    Consider a BN \(f\).
    We define the \emph{one-step percolation} of \(f\) (denoted as \(\mathcal{P}(f)\)) as follows: \(\var{\mathcal{P}(f)} = \var{f}\), \(\mathcal{P}(f)_v = f_v\) if \(f_v\) is constant, and otherwise \(\mathcal{P}(f)_v = f'_v\), where \(f'_v\) is the Boolean function obtained by substituting values of constant variables of \(f\) to \(f_v\).
\end{definition}

\begin{definition}\label{def:BN-full-percolation}
    Consider a BN \(f\).
    The \emph{percolation} of \(f\) (denoted by \(\mathcal{P}^{\omega}(f)\)) is obtained by applying \(\mathcal{P}\) starting from \(f\) until it reaches a BN \(f'\) such that \(\mathcal{P}(f') = f'\); this is always possible because the number of variables is finite.
\end{definition}

\begin{proposition}\label{prop:BN-ig-percolation}
    Consider a BN \(f\).
    One-step percolation preserves the number of variables and may reduce the number of edges, i.e., \(V(\ig{\mathcal{P}(f)}) = V(\ig{f})\) and \(E(\ig{\mathcal{P}(f)}) \subseteq E(\ig{f})\).
    Consequently, \(V(\ig{\mathcal{P}^{\omega}(f)}) = V(\ig{f})\) and\\ \(E(\ig{\mathcal{P}^{\omega}(f)}) \subseteq E(\ig{f})\).
\end{proposition}
\begin{proof}
    By definition, \(V(\ig{\mathcal{P}(f)}) = \var{\mathcal{P}(f)} = \var{f} = V(\ig{f})\).
    Let \((uv, \varepsilon)\) be an arc in \(\ig{\mathcal{P}(f)}\).
    We consider the case \(\varepsilon = \oplus\), the case \(\varepsilon = \ominus\) is similar.
    By the definition of the local influence graph (\Cref{def:local-ig}), there is a state \(x\) such that \(\mathcal{P}(f)_v(x[u \leftarrow 0]) < \mathcal{P}(f)_v(x[u \leftarrow 1])\).
    Let \(C_v\) be the set of constant variables that intersect \(\IN{f_v}\).
    By definition, \(\mathcal{P}(f)_v\) is obtained by substituting the values of the variables in \(C_v\) into \(f_v\).
    Choose state \(y\) such that \(y(v_i) = x(v_i)\) for every \(v_i \in \var{f} \setminus C_v\) and \(y(v_i) = f_{v_i}\) for every \(v_i \in C_v\).
    We have that \(f_v(y[u \leftarrow 0]) = \mathcal{P}(f)_v(x[u \leftarrow 0]) < \mathcal{P}(f)_v(x[u \leftarrow 1]) = f_v(y[u \leftarrow 1])\).
    Hence, \((uv, \varepsilon)\) is also an arc in \(\ig{f}\).
    Now we can conclude that \(E(\ig{\mathcal{P}(f)}) \subseteq E(\ig{f})\).
    It immediately follows that \(V(\ig{\mathcal{P}^{\omega}(f)}) = V(\ig{f})\) and \(E(\ig{\mathcal{P}^{\omega}(f)}) \subseteq E(\ig{f})\). \qed%
\end{proof}

\begin{proposition}[\cite{DBLP:journals/nc/KlarnerBS15,trinh2024mapping}]\label{prop:BN-percolation-async-att}
    Given a BN \(f\),
    \(\astg{\mathcal{P}(f)}\) and \(\astg{f}\) have the same set of attractors.
    Consequently, \(\astg{\mathcal{P}^{\omega}(f)}\) and \(\astg{f}\) have the same set of attractors.
\end{proposition}

\begin{theorem}\label{theo:AND-NOT-BN-num-async-att-domi-set}
    Let \(f\) be an AND-NOT BN.
    Let \(U\) be a subset of \(\var{f}\) that dominates \(\ig{f}\).
    Then the number of attractors of \(\astg{f}\) is at most \(2^{|U|}\).
\end{theorem}
\begin{proof}
    Let \(g\) be the BN such that \(g_{v_i} = v_i\) for every \(v_i \in U\) and \(g_{v_i} = f_{v_i}\) for every \(v_i \in \var{f} \setminus U\).
    By applying~\Cref{lem:BN-one-source-node-sub-graph} sequentially over variables in \(U\), we obtain \(V(\astg{g}) = V(\astg{f})\) and \(E(\astg{g}) \subseteq E(\astg{f})\).
    
    By construction, \(\ig{g}\) breaks all consistent even cycles in \(\ig{f}\).
    For any even cycle \(C\) that is strong but inconsistent in \(\ig{f}\), it is broken in \(\ig{g}\) (in case \(U\) intersects \(C\)) or it is retained in \(\ig{g}\) (in case \(U\) does not intersect \(C\)) but there is a vertex \(k_C\) that makes \(C\) inconsistent and \(U\) intersects \(k_C\).
    In addition, \(|U|\) (self) even cycles are added in \(\ig{g}\) since every \(v \in U\) is a source variable.
    Note that, since only input arcs of \(v \in U\) are removed in the transformation from \(\ig{f}\) to \(\ig{g}\), \(\ig{g}\) does not create new strong even cycles apart from the added self-cycles.
    
    For each Boolean assignment \(x\) w.r.t. variables in \(U\), we build a new BN \(g^x\) such that \(\var{g^x} = \var{f}\), \(g^x_{v_i} = x(v_i)\) for every \(v_i \in U\) and \(g^x_{v_i} = g_{v_i}\) for every \(v_i \in \var{g} \setminus U\).
    It is easy to see that the attractors of \(\astg{g^x}\) are identical to the attractors of \(\astg{g}\) that agree with the assignment \(x\), since all the variables in \(U\) are source variables of \(g\) by construction.

    By construction, \(\ig{g^x}\) breaks all self even cycles of \(\ig{g}\).
    It follows that if \(\ig{g^x}\) contains a strong even cycle \(C\), then there is a vertex \(k_C\) that makes \(C\) inconsistent and \(U\) intersects \(k_C\).

    Consider the percolation of \(g^x\), \(\mathcal{P}^{\omega}(g^x)\).
    Since \(f\) is an AND-NOT BN, \(g\), \(g^x\), and \(\mathcal{P}^{\omega}(g^x)\) are clearly AND-NOT BNs.
    By~\Cref{prop:BN-ig-percolation}, \(E(\ig{\mathcal{P}^{\omega}(g^x)}) \subseteq E(\ig{g^x})\).
    This implies that \(\ig{\mathcal{P}^{\omega}(g^x)}\) does not create new even cycles apart from the ones already in \(\ig{g^x}\).
    Let \(C\) be an even cycle that is not strong in \(\ig{g^x}\).
    Let \((u, v_1, v_2)\) be a delocalizing triple of \(C\) in \(\ig{g^x}\).
    Assume that \(C\) becomes strong in \(\ig{\mathcal{P}^{\omega}(g^x)}\).
    Then \(u \xrightarrow{\oplus} v_1\) or \(u \xrightarrow{\ominus} v_2\) needs to be removed during percolation.
    By the construction of \(\mathcal{P}^{\omega}(g^x)\), \(v_1\) or \(v_2\) becomes constant (which breaks the cycle) or \(u\) becomes constant.
    Regarding the latter case, if \(u\) becomes 0 then \(v_1\) eventually becomes 0, if \(u\) becomes 1 then \(v_2\) eventually becomes 0.
    Hence, \(C\) eventually disappears in \(\ig{\mathcal{P}^{\omega}(g^x)}\), which is a contradiction.
    This implies that \(C\) cannot become strong in \(\ig{\mathcal{P}^{\omega}(g^x)}\).
    This means that \(\ig{\mathcal{P}^{\omega}(g^x)}\) does not create new strong even cycles apart from the ones already in \(\ig{g^x}\).
    Let \(C\) be a strong even cycle in \(\ig{g^x}\).
    Then there is a vertex \(k_C \in U\) that makes \(C\) inconsistent, i.e., \(\ig{g^x}\) contains a positive path of the form \(k_C \xrightarrow{\oplus} i_1 \xrightarrow{\oplus} \dots \xrightarrow{\oplus} i_r \xrightarrow{\oplus} t_C\) (\(r \geq 0\)) and a negative path of the form \(k_C \xrightarrow{\oplus} j_1 \xrightarrow{\oplus} \dots \xrightarrow{\oplus} j_m \xrightarrow{\ominus} u_C\) (\(m \geq 0\)), such that: \(t_C\) and \(u_C\) are distinct; \(k_C \xrightarrow{\oplus} t_C\) and \(k_C \xrightarrow{\ominus} u_C\) are not arcs of \(C\); if \(m > 0\), then \(|I_{j_1}| = \dots = |I_{j_m}| = 1\) where \(I_{j}\) is the set of input arcs of vertex \(j\) in \(W\).
 %    \begin{itemize}
	% \item \(t_C\) and \(u_C\) are distinct.
	% \item If \(m > 0\), then \(|I_{j_1}| = \dots = |I_{j_m}| = 1\) where \(I_{j}\) is the set of input arcs of vertex \(j\) in \(W\).
 %    \end{itemize}
    Since \(k_C \in U\), \(g^x_{k_C}\) is a constant function.
    We have two cases as follows.
    Case 1: \(g^x_{k_C} = 0\). 
    Since \(g^x_{i_1} = k_c \land \dots\), \(\mathcal{P}^{\omega}(g^x)_{i_1} = 0\).
    By repeating this reasoning, we have \(\mathcal{P}^{\omega}(g^x)_{t_C} = 0\).
    Case 2: \(g^x_{k_C} = 1\). 
    Then \(\mathcal{P}^{\omega}(g^x)_{j_1} = 1\) because \(|I_{j_1}| = 1\).
    By repeating this reasoning, we have \(\mathcal{P}^{\omega}(g^x)_{j_m} = 1\), leading to \(\mathcal{P}^{\omega}(g^x)_{u_C} = 0\).
    In both cases, \(C\) is broken in \(\ig{\mathcal{P}^{\omega}(g^x)}\).
    Now, we can deduce that \(\ig{\mathcal{P}^{\omega}(g^x)}\) has no strong even cycle.
    
    By~\Cref{lem:AND-NOT-BN-no-even-cycle-unique-async-att}, \(\astg{\mathcal{P}^{\omega}(g^x)}\) has a unique attractor.
    By~\Cref{prop:BN-percolation-async-att}, \(\astg{g^x}\) has a unique attractor.
    There are \(2^{|U|}\) possible Boolean assignments w.r.t. variables in \(U\).
    This implies that \(\astg{g}\) has at most \(2^{|U|}\) attractors.
    By~\Cref{prop:att-sub-graph}, \(\astg{f}\) has at most \(2^{|U|}\) attractors. \qed%
\end{proof}

\begin{corollary}\label{cor:AND-NOT-BN-num-fix-domi-set}
    Let \(f\) be an AND-NOT BN.
    Let \(U\) be a subset of \(\var{f}\) that dominates \(\ig{f}\).
    Then the number of fixed points of \(f\) is at most \(2^{|U|}\).
\end{corollary}
\begin{proof}
    This immediately follows from~\Cref{theo:AND-NOT-BN-num-async-att-domi-set} and the fact that a fixed point is an asynchronous attractor. \qed%
\end{proof}

\begin{remark}
    It is easy to see that a consistent cycle is a strong cycle, but the reverse may not be true.
    It follows that if \(U\) intersects every strong even cycle of \(\ig{f}\), then it also dominates \(\ig{f}\).
    Hence,~\Cref{theo:AND-NOT-BN-num-async-att-domi-set} is stronger than~\Cref{theo:AND-NOT-BN-num-async-att-strong-even-cycle}.
\end{remark}

\begin{example}\label{exam:BN-domi-set-async-att-fix}
    Consider the BN \(f\) given as \(\var{f} = \{a, b, c, d\}\) and \(f_a = b \land d\), \(f_b = a \land \neg c\), \(f_c = d\), and \(f_d = \neg c \land d\).
    \Cref{fig:exam-BN-domi-set-async-att-fix} shows the global influence graph \(\ig{f}\).
    This graph has two even cycles \(C_1 = a \xrightarrow{\oplus} b \xrightarrow{\oplus} a\) and \(C_2 = d \xrightarrow{\oplus} d\).
    Even cycle \(C_1\) is strong but \(d\) makes it inconsistent because \(\ig{f}\) has two paths \(d \xrightarrow{\oplus} a\) and \(d \xrightarrow{\oplus} c \xrightarrow{\ominus} b\) (with \(|I_{c}| = |\{d\}| = 1\)), and \((da, \oplus), (db, \ominus) \not \in E(C_1)\).
    Even cycle \(C_2\) is clearly consistent.
    Then \(\{a, d\}\) is a minimal set that intersects every strong even cycle of \(\ig{f}\) and \(\{d\}\) is a minimal dominating set of \(\ig{f}\).
    \Cref{theo:AND-NOT-BN-num-async-att-strong-even-cycle} (resp.\ \Cref{theo:AND-NOT-BN-num-fix-cycle-delo-triple}) gives the upper bound for the number of attractors (resp.\ fixed points) of \(\astg{f}\) as \(2^2 = 4\).
    \Cref{theo:AND-NOT-BN-num-async-att-domi-set} (resp.\ \Cref{cor:AND-NOT-BN-num-fix-domi-set}) gives the upper bound for the number of attractors (resp.\ fixed points) of \(\astg{f}\) as \(2^1 = 2\).
    Indeed, \(\astg{f}\) has a unique attractor that is also a fixed point, namely 0000.
\end{example}

\begin{figure}[!t]
\centering
    \begin{tikzpicture}[node distance=1.0cm and 1.0cm, every node/.style={scale=1.0}, line width = 0.3mm]
        \node[circle, draw, minimum size=7mm] (a) [] {$a$};
        \node[circle, draw, minimum size=7mm] (b) [right=of a] {$b$};
        \node[circle, draw, minimum size=7mm] (c) [right=of b] {$c$};
        \node[circle, draw, minimum size=7mm] (d) [right=of c] {$d$};
		
        \draw[->] (a) edge [bend right=30] node [midway, below, fill=white] {$\oplus$} (b);
            
        \draw[->] (b) edge [] node [midway, above, fill=white] {$\oplus$} (a);
            
        \draw[->] (c) edge [] node [midway, above, fill=white] {$\ominus$} (b);
        \draw[->] (c) edge [bend right=30] node [midway, below, fill=white] {$\ominus$} (d);

        \draw[->] (d) edge [] node [midway, above, fill=white] {$\oplus$} (c);
        \draw[->] (d) edge [bend right=30] node [midway, above, fill=white] {$\oplus$} (a);
        \draw[->] (d) edge [loop right] node [midway, right, fill=white] {$\oplus$} (d);
    \end{tikzpicture}
    \caption{Global influence graph of the BN given in~\Cref{exam:BN-domi-set-async-att-fix}}
\label{fig:exam-BN-domi-set-async-att-fix}
\end{figure}
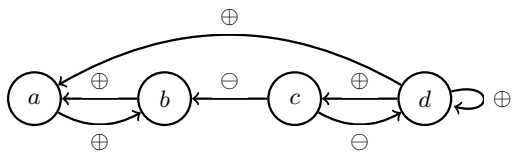

Now, we show that~\Cref{theo:AND-NOT-BN-num-async-att-domi-set} can give a better upper bound for the number of asynchronous attractors via a real-world model given as a logical regulatory graph~\cite{DBLP:journals/bmcsb/SahinFLKBMMSCTPWBA09} (see~\Cref{exam:BN-real-world}).

\begin{example}\label{exam:BN-real-world}
    Consider the logical regulatory graph for the ERBB receptor-regulated G1/S transition protein network derived from published data~\cite{DBLP:journals/bmcsb/SahinFLKBMMSCTPWBA09} (see~\Cref{fig:exam-BN-real-world}). 
    Let \(f\) be the AND-NOT BN specified by this graph.
    Note that EGF is a stimulus, thus we can consider it a source variable, i.e., \(f_{\text{EGF}} = \text{EGF}\).
    In the following we list all even cycles in \(\ig{f}\), indicating the arcs that make them inconsistent or non-strong in parentheses:
    \begin{align*}
        C_1 &= \{\text{EGF}\} \\
        C_2 &= \{\text{Akt1}, \text{IGF1R}\} \quad (\text{ErbB2\_3} \xrightarrow{\oplus} \text{Akt1}, \text{ErbB2\_3} \xrightarrow{\ominus} \text{IGF1R}) \\
        C_3 &= \{\text{Akt1}, \text{ERa}, \text{IGF1R}\} \quad (\text{ErbB2\_3} \xrightarrow{\oplus} \text{Akt1}, \text{ErbB2\_3} \xrightarrow{\ominus} \text{IGF1R}) \\
        C_4 &= \{\text{ERa}, \text{IGF1R}, \text{MEK1}\} \quad (\text{ErbB2\_3} \xrightarrow{\oplus} \text{MEK1}, \text{ErbB2\_3} \xrightarrow{\ominus} \text{IGF1R}) \\
        C_5 &= \{\text{CDK2}, \text{p27}\} \quad (\text{cMYC} \xrightarrow{\oplus} \text{CycE1} \xrightarrow{\oplus} \text{CDK2}, \text{cMYC} \xrightarrow{\ominus} \text{p27}) \\
        C_6 &= \{\text{CDK4}, \text{p27}\} \quad (\text{cMYC} \xrightarrow{\oplus} \text{CycD1} \xrightarrow{\oplus} \text{CDK4}, \text{cMYC} \xrightarrow{\ominus} \text{p27}) \\
        C_7 &= \{\text{CDK4}, \text{p21}\} \quad (\text{cMYC} \xrightarrow{\oplus} \text{CycD1} \xrightarrow{\oplus} \text{CDK4}, \text{cMYC} \xrightarrow{\ominus} \text{p21}) \\
        C_8 &= \{\text{CDK2}, \text{p27}, \text{CDK4}, \text{p21}\} \quad (\text{cMYC} \xrightarrow{\oplus} \text{CycD1} \xrightarrow{\oplus} \text{CDK4}, \text{cMYC} \xrightarrow{\ominus} \text{p21})
    \end{align*}
    It is easy to see that \(\{\text{EGF}, \text{IGF1R}, \text{CDK2}, \text{CDK4}\}\) is a minimal set that intersects every even cycle in \(\ig{f}\).
    The strong even cycles in \(\ig{f}\) are \(C_1\), \(C_5\), \(C_6\), \(C_7\), and \(C_8\).
    It is easy to see that \(\{\text{EGF}, \text{CDK2}, \text{CDK4}\}\) is a minimal set that intersects every strong even cycle in \(\ig{f}\).
    Since only \(C_1\) is consistent and \text{cMYC} makes \(C_5\), \(C_6\), \(C_7\), and \(C_8\) inconsistent, the set \(\{\text{EGF}, \text{cMYC}\}\) is a minimal dominating set of \(\ig{f}\).
    Then the result of~\cite{DBLP:journals/dam/Richard09} gives the upper bound for the number of attractors of \(\astg{f}\) as \(2^4 = 16\).
    \Cref{theo:AND-NOT-BN-num-async-att-strong-even-cycle} gives the upper bound \(2^3 = 8\), whereas~\Cref{theo:AND-NOT-BN-num-async-att-domi-set} gives the upper bound \(2^2 = 4\).
    Using the tool \texttt{biobalm}~\cite{trinh2024mapping} to compute the set of asynchronous attractors, we find that \(\astg{f}\) actually has two attractors.
\end{example}

\begin{figure}[!ht]
\centering
    \includegraphics[scale=0.4]{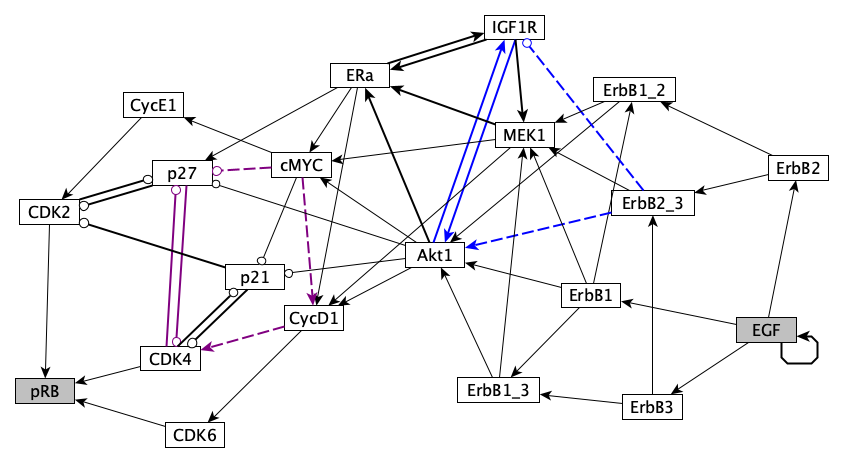}
    \caption{Logical regulatory graph considered in~\Cref{exam:BN-real-world}. Normal arrows denote positive arcs, whereas circle-tipped arrows denote negative arcs. Arcs that participate in even cycles are shown with thick continuous lines. Cycle \(C_2\) is indicated in blue and Cycle \(C_6\) is purple; the inconsistent influences on these cycles are shown with dashed lines.}
\label{fig:exam-BN-real-world}
\end{figure}
\section{Conclusion}\label{sec:conclusion}

In this work, we have established two new upper bounds for the number of asynchronous attractors of an AND-NOT BN based on strong even cycles and dominating sets in the global influence graph of the BN, respectively.
These bounds are tighter than the previous bounds.
These new findings contribute to a more comprehensive understanding of asynchronous dynamics in AND-NOT BNs with useful applications in attractor enumeration and counting, as well as in network design and control.
Beyond systems biology, they also have broader implications in the field of symbolic AI.

In the future, we plan to extend the obtained results to other classes of BNs such as OR-NOT BNs~\cite{Jarrah2010,DBLP:journals/dam/RichardR13}, AND-OR-NOT BNs~\cite{DBLP:journals/jcss/AracenaRS14}, and canalyzing BNs~\cite{DBLP:journals/tcs/LiAMAL13}.
Another direction for future work is to investigate more deeply the applications of the obtained results in the graphical analysis of abstract argumentation frameworks and finite ground normal logic programs.

\bibliographystyle{splncs04}
\bibliography{ref}

\end{document}